\documentclass[journal]{IEEEtran}
\usepackage{amsmath,amsfonts}
\usepackage{algorithmic}
\usepackage{algorithm}
\usepackage{array}
\usepackage{textcomp}
\usepackage{stfloats}
\usepackage{url}
\usepackage{verbatim}
\usepackage{graphicx}
\usepackage{cite}
\usepackage{subcaption}
\usepackage{bm, braket}
\usepackage{amsthm}
\usepackage{tikz}
\usepackage{graphicx}
\usepackage{quantikz}
\usepackage{mathtools}
\usepackage{amsthm}
\usepackage{booktabs}
\newcommand{\vdot}[2]{\bm{#1} \cdot \bm{#2}}
\newcommand{\ketbra}[2][\relax]{%
	\ifx\relax#1%
	\ket{#2}\bra{#2}
	\else
	\ket{#1}\bra{#2}
	\fi
}
\newcommand{\supp}{\operatorname{supp}}
\newcommand{\wt}{\operatorname{wt}}

\begin{document}
	\newtheorem{definition}{Definition}
	\newtheorem{lemma}{Lemma}
	\newtheorem{corollary}{Corollary}
	\newtheorem{theorem}{Theorem}
	
	\title{A Unified Framework for Optimizing Uniformly Controlled Structures in Quantum Circuits}
	
	\author{
		Chengzhuo~Xu,
		Xiao~Chen,
		Xi~Li,
		Zhihao~Liu
		and~Zhigang~Li
		\thanks{Manuscript received December xx, 2025;}
		\thanks{Chengzhuo Xu is with School of Computer Science and Engineering, Southeast University, Nanjing, 211189, China  (e-mail: 230228516@seu.edu.cn).}
		\thanks{Xiao Chen is with College of Information Engineering, China Jiliang University, Hangzhou, 310018, China (e-mail: 24a0305209@cjlu.edu.cn).}
		\thanks{Xi Li is with the School of Software, Henan University, Kaifeng 475004, China (e-mail: Lx\_nuli@yeah.net).}
		\thanks{Zhihao Liu is with School of Computer Science and Engineering, Southeast University, Nanjing, 211189, China and Key Laboratory of Computer Network and Information Integration (Southeast University), Ministry of Education, Nanjing, 211189, Jiangsu, China (email: liuzhtopic@163.com)}
		\thanks{Zhigang Li is with Nanjing Meteorology Bereau, Nanjing, 210019, Jiangsu, China (email: zhigangyangquan@163.com)}
		\thanks{(Corresponding author: Zhihao Liu.)}%
	}
	
	\markboth{Journal of \LaTeX\ Class Files,~Vol.~14, No.~8, August~2021}%
	{Shell \MakeLowercase{\textit{et al.}}: A Sample Article Using IEEEtran.cls for IEEE Journals}
	
	\IEEEpubid{0000--0000/00\$00.00~\copyright~2025 IEEE}
	
	\maketitle
	
	\begin{abstract}
		Quantum unitaries of the form ${\Sigma_{c}\ket{c}\bra{c}\otimes U_{c}}$ are ubiquitous in quantum algorithms. This class encompasses not only standard uniformly controlled gates (UCGs) but also a wide range of circuits with uniformly controlled structures. However, their circuit-depth and gate-count complexities have not been systematically analyzed within a unified framework. 
		In this work, we study the general decomposition problem for UCG and UCG-like structure. We then introduce the restricted Uniformly Controlled Gates (rUCGs) as a unified algebraic model, defined by a 2-divisible Abelian group that models the controlled gate set. This model captures uniformly controlled rotations, multi-qubit uniformly controlled gates, and diagonal unitaries.
		Furthermore, this model also naturally incorporates k-sparse version (k-rUCGs), where only a subset of control qubits participate in each multi-qubit gate.
		Building on this algebraic model, we develop a general framework. For an n-control rUCG, the framework reduce the gate complexity from ${O(n2^n)}$ to ${O(2^n})$ and the circuit depth from ${O(2^n\log n)}$ to ${O(2^n\log n/n)}$.
		The framework further provides systematic size and depth bounds for k-rUCGs by exploiting sparsity in the control space, with same optimization coefficient as rUCG, respectively.
		Empirical evaluations on representative QAOA circuits
		confirm reductions in depth and size
		which 
		highlight that the rUCG model and its associated decomposition framework unify circuits previously considered structurally distinct under a single, asymptotically optimal synthesis paradigm.
	\end{abstract}

	\begin{IEEEkeywords}
		Quantum circuit synthesis, uniformly controlled gate, circuit depth, 
		diagonal unitary operator
	\end{IEEEkeywords}
	
	\section{\label{sec:level1}Introduction}
	\IEEEPARstart{Q}{uantum} multi-controlled gates appear pervasively in quantum algorithms and subroutines. 
	Such arrays of multi-controlled gates arise whenever a computational-basis register routes execution to different program fragments: address-dependent memory access (QRAM) \cite{QRAM2008GiovannettiV,QRAM2019ParkDK,QRAM2021VerasTMLd}, oracle selection in quantum walk \cite{kUCR2025GonzalesA} or search algorithms \cite{Search2020SatohT}, block-encoding constructions in QSVT \cite{QSVT2019GilyénA,QSVT2025LiuD,QSVT2025PuramV,QSVT2025LiZ}, and many state-preparation routines \cite{QSP2019SandersYR,QSP2022ZhangXM,QSP2023MelnikovAA,QSP2024AraujoIF,QSP2024IaconisJ}. 
	These constructions all instantiate a common \emph{uniformly controlled structure}, in which a control register partitions the Hilbert space and selects target operations accordingly. 
	In practice, such uniformly controlled structures are a dominant source of circuit complexity: naively implementing an arbitrary selection over $2^n$ target operators typically requires resources that scale exponentially with the number of control qubits. 
	This observation identifies the natural technical focus of this paper: how to model and exploit algebraic restrictions of target sets in multiplexed, multi-qubit uniformly controlled structures to reduce size and depth.
	
	A convenient canonical form for multiplexed control is UCG, introduced by Möttönen \emph{et~al.}~\cite{UCR2004MöttönenM}. In its basic instantiation a UCG uses $n$ control qubits to select one of $2^n$ single-qubit target gates; its well-known decomposition attains the optimal asymptotic count for multiplexed single-qubit rotations, using $2^n-1$ CNOTs and $2^n$ single-qubit rotations. Because of this clean combinatorial structure, UCGs have been widely used as a building block in recursive synthesis and state-preparation procedures: they align naturally with binary-branching decomposition trees and provide a compact representation for control-dependent single-qubit actions~\cite{UC1G2024AllcockJ,UC1G2005BergholmV}. To conclude, the UCG may be viewed as a fundamental instance of a uniformly controlled structure. 
	
	\IEEEpubidadjcol
	Crucially, although traditional UCGs target on single-qubit gate, many practically relevant controlled-arrays are \emph{structurally} analogous to UCGs but differ in the algebraic type or dimension of the target unitaries. Representative examples include: diagonal unitary operator \cite{DUO2014NakataY,DUO2023SunX,DUO2024ZhangS}, partially controlled gatec\cite{kUCR2024LiL,kUCR2025GonzalesA}, uniformly controlled multi-qubit gate \cite{UCmG2025Bee-LindgrenM}. All these examples share the same \emph{computational-basis–conditioned selection graph} as UCGs: the control register partitions the Hilbert space into regions, and a target unitary is associated with each region.
	
	Despite this shared similar structure, synthesis and optimization techniques have largely developed in isolation for each family. Diagonal unitaries admit specialized constructions that exploit phase commutativity to reduce depth by $O(n)$~\cite{DUO2023SunX}; single-qubit UCGs admit CNOT-optimal decomposition tailored to multiplexed rotations; block-encoding and QSVT constructions use projector-based identities and spectral techniques; non-uniformly controlled operators are typically handled by bespoke conditional constructions. Because these methods rely on different normal forms and algebraic properties, simplifications obtained for one class (e.g., depth reductions for diagonal operators) do not readily transfer to other UCG-like constructions, despite sharing the same control graph. This fragmentation prevents systematic reuse of optimizations and leaves substantial parallelism and gate-count reductions unexploited.
	
	Motivated by this gap, we introduce restricted UCGs (rUCGs)—a generalized UCG family in which the target gate is restricted to one certain type. By making both the control graph and the target algebra explicit, rUCGs provide a single modeling language in which one can (i) compare decomposition complexity across families, (ii) transfer algebraic simplifications (such as commutativity-based depth reductions) from one family to another, and (iii) design parallelization strategies that exploit shared regularities. figure \ref{fig:rucg-model} summarizes how rUCGs subsume traditional UCGs and the diverse class of UCG-like structure.
	
	Building on this model, the contributions of this paper are: (1) a formal definition of rUCGs and algebraic criteria that characterize useful target restrictions; (2) decomposition methods that reduce circuit size and depth; and (3) application-driven benchmarks demonstrate that optimizations known for diagonal or commuting targets. 
	
	\begin{figure}[!t]
		\centering
		\includegraphics[width=\columnwidth]{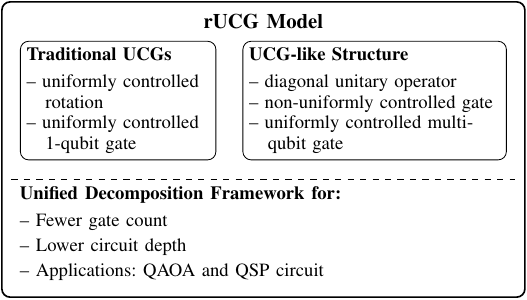}
		\caption{
			Overview of the proposed rUCG model that unifies traditional UCGs and UCG-like operator families within a single model, together with the goal of unified decomposition framework for size and depth.
		}
		\label{fig:rucg-model}
	\end{figure}
	
	The paper is organized as follows. Section~\ref{sec:level2} introduces definitions and preliminaries. Section~\ref{sec:level3} formalizes the rUCG model. Sections~\ref{sec:level} and~\ref{sec:level4} present decomposition of optimized size and depth. Section~\ref{sec:level5} benchmarks representative applications, and Section~\ref{sec:level7} concludes.
	
	\section{\label{sec:level2}Preliminaries}
	
	\subsection{Uniformly Controlled Gate (UCG)}
	
	UCGs constitute a class of multi-control operators determined entirely by the state of control qubits, referred to as the \emph{control state}. For an n-qubit control register, the computational basis $\{|c\rangle\}_{c=0}^{N-1}$ (with $N=2^n$) enumerates all possible control states, each of which selects a corresponding unitary operator on the $m$-qubit target system. As illustrated in figure \ref{fig:ucg}, a UCG $L_n[\mathbf{U}]$ is decomposed as
	\begin{equation}
		L_n[\mathbf{U}] = \prod_{i=0}^{N-1} C_n^i[U_i],
		\label{eq:ucg}
	\end{equation}
	where $\mathbf{U} = \{U_i\}$ is a sequence of target gates. Each factor $C_n^i[U_i]$ is a conditioned unitary activated only when the control register is in state $|i\rangle$. We refer to $U_i$ as the \emph{target operator} under control state $|i\rangle$.
	
	\begin{figure}[!t]
		\centering
		\includegraphics[width=\linewidth]{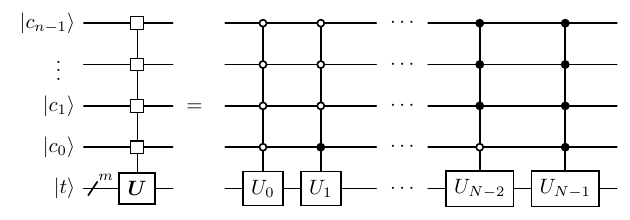}
		\caption{Uniformly controlled gate (UCG) $L_n[\mathbf{U}]$ implemented as a sequence of $N=2^n$ conditioned unitaries $C_n^i[U_i]$. Each $U_i$ acts on the $m$-qubit target register $|t\rangle$ when the control register is in control state $|i\rangle$.}
		\label{fig:ucg}
	\end{figure}
	
	A representation especially suitable for analysis expresses the UCG as a control-state–indexed block-diagonal operator:
	\begin{equation}
		L_n[\bm{U}] = \sum_{c=0}^{N-1} |c\rangle\langle c| \otimes U_c .
		\label{eq:ucg2}
	\end{equation}
	
	\subsection{Optimization of UCGs}
	
	Optimization becomes difficult when the targets are completely arbitrary. However, many relevant real-world examples demonstrate that an optimizable UCG may meet the following conditions:
	
	\begin{itemize}
		\item \textbf{Commutability of the target.}  
		All target operators belong to the same operator family, often forming a commuting group, ensuring simplified decomposition and synthesis.
		
		\item \textbf{Density of controls.} 
		Multi-controlled gate count is large enough for optimization to be effective.
		\begin{itemize}
			
			\item \textbf{Global density.}  
			Dense over full control-state space from $0$ to $2^n\!-\!1$.
			
			\item \textbf{Localized density.}
			Dense over a control-state subspace that follows a certain pattern.
		\end{itemize}
	\end{itemize}
	
	To ensure generality and abstraction, we treat $U$ as a symbolic placeholder for an elementary gate whose precise structure is not fixed during decomposition.
	
	\begin{definition}
		\label{def:Uprimitive}
		We work in the \emph{U-primitive} model, together with basic gate set (including CNOT and single-qubit gates), constituting indivisible the unit of UCG decomposition.
		Given the target operator family $\mathcal{U}$, U-primitive consists of
		\begin{equation*}
			\mathcal{G}_{CU} \;=\; \{ \, \text{controlled-}U \;|\; U\in\mathcal{U} \,\} \;\cup\; \mathcal{U}.
		\end{equation*}
		Under this model we measure circuit cost by the size and depth of occurrences of elements in $\mathcal{G}_{CU}$ and basic gates. 
		
		Occasionally, gates algebraically derived from $\mathcal{G}_{CU}$ 
		through diagonalization are needed.
		We refer to this enlarged family as the U-closed family.
	\end{definition}
	
	For simplicity, we will use U and CU to indicate which category of gate is used for decomposition.

	\subsection{Special Cases of Uniformly Controlled Structures}
	
	UCGs encompass several representative configurations that illustrate both the classical scope of UCGs and the boundary cases that motivate our generalized rUCG model. This discussion focuses on three representative configurations of UCGs:
	
	\begin{itemize}
		\item Uniformly controlled rotations, a classical simple case.
		\item Diagonal unitary operators, a case extending target operators.
		\item k-sparse UCGs, a sparse controlled case.
	\end{itemize}
	
	\paragraph{Uniformly controlled rotations}
	The familiar uniformly controlled rotations corresponds to the case where each $U_c$ is a single-qubit rotation whose angle depends on $c$. In this regime, the CNOT cascade \ provides an $O(2^n)$-gate decomposition driven purely by uniformity \cite{UCR2004MöttönenM}. This instance reflects the simplest and most commonly used UCG structure:
	\begin{equation}
		L_n[\mathbf{R}]
		= \sum_{c=0}^{N-1} |c\rangle\langle c| \otimes R(\theta_c),
	\end{equation}
	where $\mathbf{R}$ denotes a sequence of rotations, and $R(\theta_c)$ denotes a single-qubit rotation about a fixed axis, typically chosen from $\{R_y, R_z\}$.
	
	\paragraph{Diagonal unitary operators}
	Diagonal unitaries can be viewed as uniformly controlled phases, with the target operator degenerated into scalars:
	\begin{equation}
		L_n(\mathbf{\Lambda})
		= \sum_{c=0}^{N-1} e^{j\chi_c}\ket{c}\bra{c}.
	\end{equation}
	
	\paragraph{k-sparse UCGs}
	A further extension is the \emph{k-sparse UCG}, where each target operator depends only on a k-size subset of control qubits:
	\begin{equation}
		L_{n,\le k}[\mathbf{U}]
		= 
		\mathrm{LU}_{n,k}=\prod_{i} C_{Q_i}[U_i],
	\end{equation}
	where $Q_i, U_i$ are control qubits and target gate of $i$-th gate, and $|Q_i|\le k$.
	This structure effectively generalizes UCGs to allow for non-uniform controlled pattern. 
	
	As a specific case, we use
	$\mathrm{LU}_{n,k}^{=k}$
	to denote the k-weight UCG, where the number of control qubits is fixed at $k$.
	
	\section{\label{sec:level3}Restricted Uniformly Controlled Gates}
	
	The goal of this section is to introduce a general algebraic abstraction that captures all of these optimizable cases under a single model.
	
	\subsection{Definition and Algebraic Model}
	
	Consider a UCG acting on a n control qubits, with target operators restricted to a fixed set of types, denoted as $\mathcal{U}$.
	
	\begin{definition}[Restricted UCG, rUCG]
		\label{def:rucg}
		A UCG is a \emph{rUCG} if the target operator family $\mathcal{U}$  forms a subgroup satisfying:
		\begin{enumerate}
			\item[(R1)] \textbf{Abelian closure:} $U_a U_b = U_b U_a$ for all $U_a,U_b\in\mathcal{U}$;
			\item[(R2)] \textbf{2-divisibility:} for every $U\in\mathcal{U}$ there exists $V\in\mathcal{U}$ such that $V^2 = U$.
		\end{enumerate}
	\end{definition}
	
	These conditions deliberately encompass a broad class of optimizable UCG instances, such as:
	\begin{itemize}
		\item uniformly controlled rotations \cite{UCR2004MöttönenM};
		\item diagonal-unitary operators\cite{DUO2014NakataY};
		\item uniformly controlled commuting Hamiltonian exponentials $e^{j\theta H}$ with mutually commuting $H$\cite{UCH2019LowGH};
		\item uniformly controlled modular adders (mod an odd number) \cite{UCADD2004MeterRV}.
	\end{itemize}
	
	\subsection{\label{sec:rucg-represent}Representation and Vector Embedding}
	
	\begin{definition}
		Let $(D_2,+)$ be a 2-divisible Abelian group.  
		A gate family $\mathcal{U}$ is said to be \emph{defined over $D_2$} if there exists a group homomorphism 
		\(
		U(\,\cdot\,):D_2 \to \mathcal{U}
		\)
		(still sometimes we use the single $U$ to denote a standalone unitary matrix for simplicity).
	\end{definition}
	
	\begin{corollary}
		For any $a,b\in D_2$, the homomorphism property gives
		\[
		U(a+b)=U(a)U(b)
		\quad\text{and}\quad
		U(a/2)=\sqrt{U(a)}.
		\]
	\end{corollary}
	
	\begin{proof}
		Both identities follow immediately from $U$ being a homomorphism and $D_2$ being 2-divisible. 
	\end{proof}
	
	Once the relationship between the unitary gates and the group elements is established, a UCG is uniquely determined by a vector. Simplify the gate denotation $\mathrm{LU}(\bm{\chi}) := L_n \left[ \{U(\bm{\chi})\} \right]$, where $\chi$ is termed the \emph{target vector}, an $N$-dimension vector with each element $\chi_c$ being from $D_2$. Thus, Eq.~\eqref{eq:ucg2} can be reformulated as:
	\begin{equation}
		\text{LU}\left( \bm{\chi } \right) =\sum\limits_{c=0}^{N-1}{\left| c \right\rangle \left\langle  c \right|\otimes U\left( {{\chi }_{c}} \right)},  
	\end{equation}
	as is shown in figure \ref{fig:rucg}.
	
	When the target operator is a phase, that is, for the diagonal unitary operator, it can be specifically expressed as:
	
	\begin{equation}
		\Lambda\left( \bm{\chi} \right) =\sum\limits_{c=0}^{N-1} e^{ {j{\chi}_{c}} } \ketbra{c}.
		\label{eq:duo}
	\end{equation}
	
	\begin{figure}[!t]
		\centering
		\includegraphics[width=\linewidth]{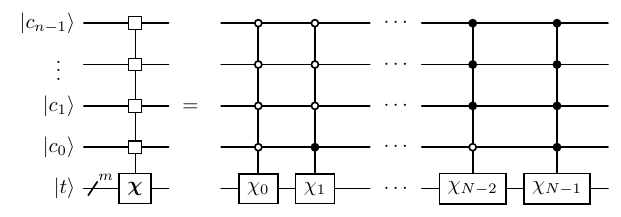}
		\caption{Restricted uniformly controlled gate(rUCG) $\mathrm{LU}(\bm{\chi})$, where $\bm{\chi}$ is called a target vector determining the target operators. For brevity, $\chi_i$ denotes $U(\chi_i)$ in a gate.}
		\label{fig:rucg}		
	\end{figure}
	
	\subsection{Local-Control Structure and k-sparse rUCG}
	Next, we append the control locality on rUCGs.  
	Throughout this subsection, let n control qubits labeled by the index set $[n]=\{1,\dots,n\}$.
	
	\begin{definition}
		\label{def:weight-sparse}
		A controlled gate $C_Q(U)$ is said to be \emph{k-weight control} if its set of control qubits
		\(
		Q\subseteq[n]
		\)
		satisfies $|Q|= k$, and all controls on the qubits in $Q$ are \emph{positive controls} (black dots).
		Additionally, if $|Q|\leq k$, $C_Q(U)$ is said to be \emph{k-sparse control}.
	\end{definition}
	
	\begin{definition}[Standard k-rUCG]
		\label{def:standard-krucg}
		An operator is a \emph{standard k-rUCG} if it can be written as a product of controlled gates
		\begin{equation*}
			\mathrm{LU}_{n,k}=\prod_{i} C_{Q_i}[U(\mu_i)],
		\end{equation*}
		where $Q_i$ and $U(\mu_i)$ is control qubits and target gate of $i$-th controlled gate $C_{Q_i}[U(\mu_i)]$, which satisfies:
		\begin{enumerate}
			\item $U(\cdot)$ maps to the $2$-divisible Abelian subgroup $\mathcal{U}$;	
			\item each control is k-sparse, which satisfies $|Q_i|\le k$;
			\item all controls in $C_{Q_i}[U(\mu_i)]$ are positive (black dots).
		\end{enumerate}
	\end{definition}
	
	\begin{definition}[General k-rUCG]
		\label{def:general-krucg}
		A operator is a \emph{general k-rUCG} if it can be written in the form of Definition~\ref{def:standard-krucg}, except that some controls are allowed to be negative controls (white dots).
	\end{definition}
	
	We now show that white dots add no expressive power.
	
	\begin{lemma}
		\label{lem:white-to-black}
		Any general k-rUCG containing white controls can be standardize as a product of controlled gates having only positive controls or zero controls.
	\end{lemma}
	
	\begin{proof}
		Let $C_{\bar q}(U)$ denote the single-qubit-controlled gate that applies $U$ to the target iff control qubit $q$ is in state $\ket{0}$ (white-dot), and let $C_{q}(U)$ denote the positive-controlled (black-dot) gate applying $U$ when $q$ is $\ket{1}$. Then for any zero-controlled-$U$ there holds
		\begin{equation*}
			C_{\bar q}[U] \;=\; \big(I_{q}\otimes U\big)\; C_{q}\big[U^\dagger\big].
		\end{equation*}
	\end{proof}
	
	The following observation connects the structural sparsity of k-sparse controls with the vector representation of rUCGs.
	
	\begin{theorem}
		\label{thm:krucg-subclass}
		If an operator is a k-rUCG, then it is also an rUCG.
	\end{theorem}
	
	\begin{proof}
		Let $Q\subseteq[n]$ be a control subset with $|Q|\le k$, and let
		\begin{equation*}
			\label{eq:delta_qc}
			\delta_Q(c)=
			\begin{cases}
				1, & \text{if $Q \subset \supp(c)$},\\[2mm]
				0, & \text{otherwise}.
			\end{cases}
		\end{equation*}
		Here $\supp(c)=\{q~|~c_q=1\}$ represents the set of bits with value 1 in binary.	
		Then any $Q$-controlled unitary can be written as
		\begin{equation*}
			C_Q(U)=\sum_{c=0}^{N-1} \ketbra{c}\otimes \big(\delta_Q(c) U + \overline{\delta_Q(c)} I\big),
		\end{equation*}
		namely, as a product of n–controlled gates.
		Therefore a k-rUCG—which is by definition a product of such
		$C_{Q_i}(U_i)$ with $|Q_i|\le k$—is an rUCG.
	\end{proof}
	
	\begin{corollary}
		\label{cor:krucg-vector}
		Every k-rUCG admits an rUCG-style vector representation
		$\mathrm{LU}(\boldsymbol{\chi})$.
		
		\begin{proof}
			By Theorem~\ref{thm:krucg-subclass}, any k-rUCG is a product of
			n-controlled gate, each $C_{Q_i}[U(\mu_i)]$ corresponding to a target vector
			$\boldsymbol{\chi}^{(i)}$ whose only nonzero entry encodes a target gate and whose other entries are the identity operator.
			
			\begin{equation}
				\label{eq:target-vector-kg}
				\boldsymbol{\chi}^{(i)} = \mu_i \cdot \big(\delta_{Q_i}(0), \delta_{Q_i}(1), \ldots, \delta_{Q_i}(N-1)\big)^T.			
			\end{equation}
			
			Since multiplication of commuting target operators corresponds to
			addition in the target-vector space, the product of all these factors is
			exactly
			\begin{equation}
				\label{eq:target-vector-com}
				\boldsymbol{\chi}
				= \sum_{i} \boldsymbol{\chi}^{(i)} ,
			\end{equation}
			which uniquely determines the resulting rUCG in the form
			$\mathrm{LU}(\boldsymbol{\chi})$.
			Thus every k-rUCG possesses a well-defined vector representation.
		\end{proof}
	\end{corollary}
	
	\begin{figure}[!t]
		\centering
		\begin{minipage}{\linewidth}
			\centering
			\includegraphics[width=\linewidth]{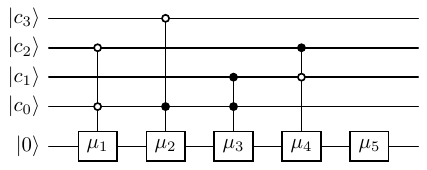}
			\caption*{(a)}
			\label{fig:k-rucg1}
		\end{minipage}
		
		\vspace{1em} 
		
		\begin{minipage}{\linewidth}
			\centering
			\includegraphics[width=\linewidth]{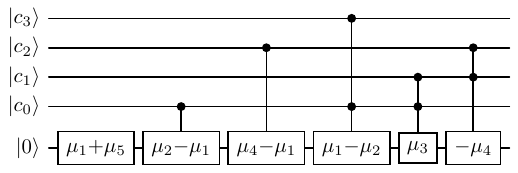}
			\caption*{(b)}
			\label{fig:k-rucg2}
		\end{minipage}
		
		\caption{Example of a k-rUCG for $n=4, k=2$ (a) general form. (b) standard form. Finally converted into vector form LU\((\mathbf\chi)\), where  \(\mathbf{\chi} = (\mu_1+\mu_5, \mu_2+\mu_5, \mu_1+\mu_5, \cdots, \mu_3+\mu_5)\).}
		\label{fig:k-rucg}
	\end{figure}

	\section{\label{sec:level}Decomposition of rUCG with Fewer Gates}
	\subsection{\label{sec:size-trans}Frequency-Domain Transformation}
	The circuit complexity of the rUCG primarily stems from multi-controlled gates. An n-controlled gate can be decomposed into $O(n)$ size and $O(log n)$ depth\cite{nCtrlDepth2022daSilvaAJ,nCtrlDepth2024NieJ}. The main reason for this is that the control conditions for each qubit are connected via ``AND'' operations. To address this issue, we propose a frequency-domain transformation for the rUCG.
	
	Let \(\mathrm{LU}(\boldsymbol{\chi})\) denote an rUCG with target unitaries \(U(\chi_i)\) associated with control states \(\ket{i}\).  
	Apply the unnormalized Walsh--Hadamard transform \cite{Walsh1969ShanksJL}  to \(\boldsymbol{\chi}\),
	\begin{equation}
		\label{eq:walsh-hadamard}
		X_{\omega}=\sum_{i=0}^{N-1}(-1)^{\boldsymbol{\omega}\cdot\mathbf{i}}\,\chi_i,
	\end{equation}
	and introduce a new vector \(\mathbf{Y}\):
	\begin{equation}
		\label{eq:frequecy-vector}
		Y_0=\frac{1}{N}\sum_{\omega=0}^{N-1}X_\omega,\qquad
		Y_\omega=-\frac{2}{N}X_\omega\quad (\omega\neq 0),
	\end{equation}
	
	referred as \emph{frequency vector}.
	
	\begin{lemma}
		The rUCG admits the decomposition:
		\begin{equation}
			\label{eq:freq-decomp}
			\begin{aligned}
				\mathrm{LU}(\bm{\chi})& =  \left( I^{\otimes n} \otimes U(Y_{0}) \right) \cdot \\
				& \prod_{\omega = 1}^{N-1} \Biggl( \sum_{c = 0}^{N-1} \ketbra{c} \otimes 
				\left( \delta_{\bm{c} \cdot \bm{\omega}, 1} U(Y_{\omega}) + \overline{\delta_{\bm{c} \cdot \bm{\omega}, 1}} I \right) \Biggr),
			\end{aligned}
		\end{equation}
		where ${\delta }_{c,i}$ is 1 when $c=i$, and 0 otherwise. 
		Consequently, the original condition \(c=i\), or \(
		(c_0 = i_0) \land (c_1 = i_1) \land \cdots \land (c_{n-1} = i_{n-1})
		\) is replaced by the condition  
		\(
		\mathbf{c}\cdot\boldsymbol{\omega}
		= c_0\omega_0\oplus c_1\omega_1\oplus\cdots\oplus c_{n-1}\omega_{n-1}
		\). 
		
		Here, two equivalent notations are used for a control basis state.  
		One is the numerical form $c$, whose binary expansion is  
		$c=c_{n-1}\cdots c_1 c_0$;  
		the other is the vector form $\boldsymbol{c}=(c_0,c_1,\dots,c_{n-1})$.  
		These two representations are interchangeable except when clarity is needed in calculations or as subscripts.
	\end{lemma}

	\begin{proof}
		We outline the derivation for completeness.
		
		\begin{enumerate}
			\item The standard form of an rUCG can also be written as
			\begin{equation}
				\label{eq:rucg-standard}
				\mathrm{LU}(\boldsymbol{\chi})
				=
				\prod_{i=0}^{N-1}
				\sum_{c=0}^{N-1}
				\ket{c}\!\bra{c}\otimes
				\left(
				\delta_{c,i}U(\chi_i)+\overline{\delta_{c,i}}I
				\right).
			\end{equation}
			
			\item Apply the inverse Walsh--Hadamard transform
			\begin{equation}
				\chi_i=\frac{1}{N}\sum_{\omega=0}^{N-1}
				(-1)^{\boldsymbol{\omega}\cdot\mathbf{i}}X_\omega.
			\end{equation}
			
			\item Substitute this expression into Eq.~\eqref{eq:rucg-standard} and separate the \(\omega=0\) and \(\omega\neq 0\) contributions.
			
			\item Introduce frequency vector \(\mathbf{Y}\) to simplify coefficients, yielding Eq.~\eqref{eq:freq-decomp}.
		\end{enumerate}
	\end{proof}
	
	To summarize, the transformation presented in this section converts the rUCG from target vector under original control states to the \emph{frequency vector} within the frequency domain. This shift in the control state simplifies subsequent operations performed upon it. For convenience, all \emph{later references to the control states} in this article will pertain to the transformed ones in the frequency domain.
	
	To implement the circuit of $\mathrm{LU}(\bm{\chi})$, the remaining issue is how to efficiently compute, store, and transfer the control state $\ket{\bm{c} \cdot\boldsymbol{\omega}}$.
	
	\subsection{Traversal of Control States and Activation Mechanism}
	\label{sec:size-traversal}
	This subsection develops a traversal-based mechanism that uses only CNOT gates and a single designated control qubit to sequentially activate each control condition.
	
	\paragraph*{Traversal Sequence.}
	Let $S=\{S_0,S_1,\dots,S_{N-1}\}$ be a sequence of circuits acting on the
	control register.
	Each $S_i$ consists solely of CNOT gates (with an optional ancilla) and is
	designed to \emph{traverse} a predefined ordering of the  control states
	$\{\bm{c}\cdot\boldsymbol{\omega}\}_{\omega=0}^{N-1}$.
	
	Starting from an initial state
	\(
	\ket{s^{(0)}}=\ket{0}^{\otimes n},
	\)
	the sequence evolves as
	\begin{equation}
		\ket{s^{(i+1)}} = S_i \ket{s^{(i)}}, 
	\end{equation}
	such that each $\ket{s^{(i)}}$ corresponds to a distinct traversal step.
	After completing all $N$ steps, the sequence returns to the initial state:
	\(
	\ket{s^{(N)}} = \ket{s^{(0)}}.
	\)
	
	\paragraph*{Activation Qubit}
	A crucial property of the traversal sequence is its \emph{activation mechanism}.
	For each intermediate state $\ket{s^{(i)}}$, a pre-designated
	\emph{activation qubit}—denoted by $\ket{s^{(i)}_{\mathrm{a}}}$—stores the value of
	the current control condition:
	\begin{equation}
		s^{(i)}_{\mathrm{a}} =
		\bm{c}\cdot\boldsymbol{\omega}^{(i)} =
		c_0\omega^{(i)}_0 \oplus c_1\omega^{(i)}_1 \oplus \dots \oplus c_{n-1}\omega^{(i)}_{n-1}.
	\end{equation}
	
	The sequence is designed such that
	\begin{equation}
		s^{(i+1)}_{\mathrm{a}}=1
		\quad\Longleftrightarrow\quad
		\bm{c}\cdot\boldsymbol{\omega}^{(i+1)}=1.
	\end{equation}
	which allows us to activate any controlled operator using only a \emph{single}
	control qubit, instead of relying on a multi-qubit AND.
	
	\paragraph*{Composite Step Operators}
	At each traversal step $i$, we apply a composite operation $\mathrm{CU}_i$
	consisting of:
	
	1. A single-control operation on the target register, triggered by the
	activation qubit.
	
	2. Followed by the traversal update.
	
	\begin{equation}
		\label{eq:rucg-step}
		\begin{aligned}
			\mathrm{CU}_i &: ~
			\ket{s^{(i)}}\ket t \to 
			\ket{s^{(i)}}\!\left(
			\delta_{s^{(i)}_a,1} U(Y_{\omega^{(i)}})
			+ \overline{\delta_{s^{(i)}_a,1}} I
			\right)\ket t,
			\\
			S_i  &: ~
			\ket{s^{(i)}} \to \ket{s^{(i+1)}}.
		\end{aligned}
	\end{equation}
	
	Thus each $\mathrm{CU}_i$ is a lightweight operation driven only by the
	single activation qubit.
	
	\paragraph*{Final Circuit Form}
	Combining all steps of traversal and activation, the rUCG admits the following
	circuit decomposition:
	\begin{equation}
		\label{eq:rucg-circuit}
		\begin{aligned}
			\mathrm{LU}(\bm{\chi})
			&=
			\biggl(
			\prod_{i=N-1}^{1}
			\bigl(S_i\otimes I^{\otimes m}\bigr)\,
			CU_{s^{(i)}_{\mathrm{a}}}(Y_{\omega^{(i)}})
			\biggr)
			\\
			&\quad\cdot
			\bigl(S_0\otimes I^{\otimes m}\bigr)
			\bigl(I^{\otimes n}\otimes U(Y_0)\bigr),
		\end{aligned}
	\end{equation}
	where each $CU_{s^{(i)}_{\mathrm{a}}}$ is a \emph{single-control} gate by activation qubit and the
	sequence $\{S_i\}$ contains only CNOT gates.  
	This circuit is equivalent to inserting each controlled operator at precisely
	the traversal step corresponding to its control state, as illustrated in
	figure \ref{fig:GP}.
	
	\subsection{Implementation for rUCGs}
	\label{sec:size-rucg}
	
	Selecting a suitable traversal algorithm can realizes the sequence \(S=\{S_0,\dots,S_{N-1}\}\)
	using only CNOTs.
	Substituting $S$ into Eq.\eqref{eq:rucg-circuit} yields the final decomposition.
	
	\begin{figure}[!t]
		\centering
		
		\begin{minipage}{\linewidth}
			\centering
			\includegraphics[width=\linewidth]{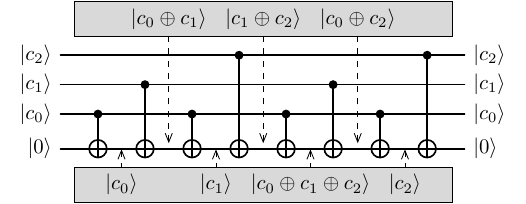}
			\caption*{(a)}\label{fig:GP1}
		\end{minipage}
		
		\vspace{1.2em}
		
		\begin{minipage}{\linewidth}
			\centering
			\includegraphics[width=\linewidth]{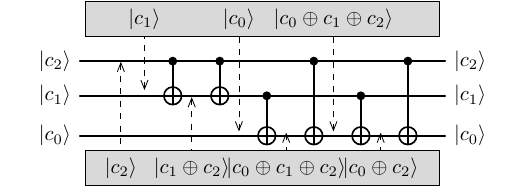}
			\caption*{(b)}\label{fig:GP2}
		\end{minipage}
		
		\vspace{1.2em}
		
		\begin{minipage}{\linewidth}
			\centering
			\includegraphics[width=\linewidth]{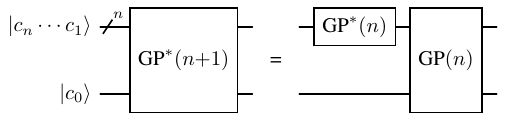}
			\caption*{(c)}\label{fig:GP3}
		\end{minipage}
		
		\caption{Gray Path GP(n) for traversing all control states.  
			(a) GP($3$) with an ancilla.  
			(b) GP$^*$($3$) without ancilla.  
			(c) Recursive construction for GP$^*$(n).}
		\label{fig:GP}
	\end{figure}
	
	\begin{theorem}[Gate count of rUCGs]
		\label{thm:rucg-gatecount}
		A rUCG can be decomposed to
		$2^n-1$ CU,
		a single U
		and $2^n-2$ CNOTs.
	\end{theorem}
	
	\begin{proof}
		We designed two variants for implement $S$ based on Gray code: GP$(n)$ with an ancilla, and GP$^*(n)$ without ancilla.
		
		\paragraph{GP$(n)$ with an ancilla}
		A binary reflected Gray code
		\(\{\boldsymbol{\omega}^{(0)},\boldsymbol{\omega}^{(1)},\dots,\boldsymbol{\omega}^{(N-1)}\}\)
		enumerates all \(n\)-bit vectors such that consecutive vectors differ in
		exactly one bit~\cite{CombGray1997SavageC}. 
		
		Allocate one ancilla \(a\) initialized to \(\ket{0}_a\) as the activation qubit. For each step \(i\),
		\begin{equation}
			s^{(i)}_{\mathrm{a}} \coloneqq \bm{c}\cdot\boldsymbol{\omega}^{(i)}
			= \bigoplus_{l:\,\omega^{(i)}_l=1} c_l,
		\end{equation}
		successive \(\boldsymbol{\omega}^{(i)}\) differ
		by a single bit flip at some index \(l\), updating the ancilla from
		\(s^{(i)}_{\mathrm{a}}\) to \(s^{(i+1)}_{\mathrm{a}}\) requires a single
		CNOT controlled on \(c_l\) and targeting \(a\):
		\begin{equation}
			S_i^{(\text{anc})} = \mathrm{CNOT}(c_{l}\to a).
		\end{equation}
		
		Due to the properties of Gray code, GP$(n)$ uses $2^n$ CNOT gates.
		
		\paragraph{GP$^*(n)$ without ancilla}
		GP$^*(n)$ is designed based on GP$(n)$. Divide all control states into two parts: those that do not contain the last qubit and those that do contain the last qubit. Then apply the recursive method to the former and apply $GP(n)$ to the latter (treating the last qubit as an ancilla).
		The recursive decomposition of \(\mathrm{GP}^*(n)\) is shown schematically in figure \ref{fig:GP}(b)--(c).
		
		Due to the recursive construction based on GP$(n)$, GP$^*(n)$ uses $2^n-2$ CNOT gates.
		
		\paragraph{Integration}
		By combining \(\mathrm{GP}^*(n)\) with the sequence of single-control operations
		\(CU_{s^{(i)}_{\mathrm{a}}}(Y_{\omega^{(i)}})\), we obtain the full
		circuit decomposition Eq.~\eqref{eq:rucg-circuit}. figure \ref{fig:rucg-final} shows the integrated circuit, which proves Theorem~\ref{thm:rucg-gatecount}.
	\end{proof}
	
	\begin{figure*}[t]
		\centering
		\includegraphics[width=\linewidth]{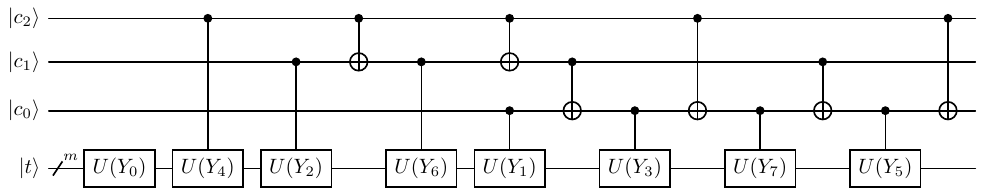}
		\caption{
			Decomposition of rUCG with optimal number of gates for $n=3$, within GP$^*(3)$.
		}
		\label{fig:rucg-final}
	\end{figure*}
	
	\subsection{Implementation for \(k\)-sparse rUCGs}
	\label{sec:size-krucg}
	
	For \emph{\(k\)-sparse rUCG}, not all control states need to be traversed. In Eq.~\eqref{eq:rucg-standard}, if a controlled gate is an identity $I$, then the corresponding control state does not need to be traversed. 
	
	A critical clarification is necessary for the following concept: while it resembles Definition~\ref{def:weight-sparse}, the two are subtly different but distinct:
	
	\begin{definition}
		\label{weight-sparse-state}
		For each traversal step of control state $s^{(i)}_{\mathrm{a}}$ obtained in Section~\ref{sec:size-traversal}, a state is k-weight if its the Hamming weight $\wt(s^{(i)}_{\mathrm{a}})=k$, while it is k-sparse if $\wt(s^{(i)}_{\mathrm{a}})\leq k$.
	\end{definition}
	
	Namely, the weight of control is the number of control qubits, and the weight of a control state is the number of 1s (black dot) in its binary value.
	
	\begin{lemma}
		\label{lem:k-sparse}
		The control states derived for the k-rUCG via transformation in Section~\ref{sec:size-trans} are all k-sparse. That is, for each $\wt(s^{(i)}_{\mathrm{a}}) \leq k$. 
		
		This unifies k-sparse control and k-sparse control state, with the standard k-rUCG form using no more than k control qubits being equivalent to the rUCG representation of the frequency domain control state with weights no greater than k.
		
		Consequently, the length of the $S$ is
		\begin{equation*}
			|S| \leq \sum_{i=0}^k \binom{n}{i}.
		\end{equation*}
	\end{lemma}
	
	\begin{proof}
		See Appendix~A in the supplementary material.
	\end{proof}

	For a different scale of $S$, another traversal method is needed for the sparse control states. By definition a \(k\)-rUCG acts nontrivially only for those
	\(\boldsymbol{\omega}\) with $\wt(\omega) \leq k$. In the study of related weight-range Gray codes\cite{RangeGray2018GregorP, RangeGray2022MützeT}), this type of problem is called \(Q_{n,[0,k]}\), which means traversing only binary codes with weights from 0 to $k$, or we call \emph{k-sparse Code} according to Definition~\ref{def:weight-sparse}.
	
	\begin{theorem}[Gate count of \(k\)-rUCGs]
		\label{thm:krucg-gatecount}
		A k-rUCG can be decomposed to $\sum_{i=1}^{k} \binom{n}{i}$ CU,
		a single U and less than $2\sum_{i=0}^{k} \binom{n}{i}$ CNOTs.
	\end{theorem}
	
	\begin{proof}
		We designed two variants for implement $S$ based on k-sparse Gray code: GP$(n, k)$ with an ancilla, and GP$^*(n, k)$ without ancilla.
		
		\paragraph{GP$(n, k)$ with an ancilla}
		\cite{RangeGray2022MützeT}) guarantee that for the set of weights
		\([0,k]\) there exists a near-optimal closed walk (a cycle or near-cycle)
		that visits every vertex in \(Q_{n,[0,k]}\) with successive vertices
		differing by one or a two constant number of bit flips. We therefore
		can construct a sequence \( \{\boldsymbol{\omega}^{(i)}\} \subset
		Q_{n,[0,k]}\) such that consecutive elements are (nearly) adjacent in \(Q_n\).
		
		Similar to the full-rUCG case we store
		\(s^{(i)}_{\mathrm{a}}=\bm{c}\cdot\boldsymbol{\omega}^{(i)}\) in an ancilla
		\(a\), with each transition needs one or two CNOTs, yields the following upper bound on CNOT count:
		
		\begin{equation}
			\label{eq:ksparse-cost}
			\mathrm{CNOTs}\left(\mathrm{GP}(n,k)\right) \; \le \; 2\sum_{i=0}^{k} \binom{n}{i}.
		\end{equation}
		
		\paragraph{GP$^*(n, k)$ without ancilla}
		Ancilla-free schemes for the \(k\)-sparse case can also be obtained by recursive construction, A \(\mathrm{GP}^*(n,k)\) can be divided into a \(\mathrm{GP}^*(n-1,k)\) and \(\mathrm{GP}(n-1,k-1)\). It is easy to prove that the number of CNOT gates required by \(\mathrm{GP}^*(n,k)\) also satisfies Eq.~\eqref{eq:ksparse-cost}.
		
		Both GP$(n, k)$ and GP$^*(n, k)$ proves the CNOT count in Theorem~\ref{thm:krucg-gatecount}, with CU and U count equals the count of control states.
		
	\end{proof}
	
	\subsection{Section Discussion}
	\label{sec:gate-count-summary}
	This work have presented an optimized decomposition framework for rUCG. The main contributes are:
	(i) a frequency domain reformulation via the
	Walsh--Hadamard transform (Section~\ref{sec:size-trans});
	and (ii) an efficient traversal mechanism for these control states
	(Section~\ref{sec:size-rucg}), which implements the decomposition in
	circuit form Eq.~\eqref{eq:rucg-circuit} using only CNOTs and
	single-controlled gates; finally (iii) Integrate into circuit construction logic of traversal activation and controlled gate insertions (Section~\ref{sec:size-traversal}).
	
	The proposed framework opens several practical directions: concrete
	routing algorithms for \(\mathrm{GP}^*(n)\) with provable CNOT counts;
	explicit synthesis methods for single-control target gates within common
	hardware gate sets.
	
	Beyond the general rUCG model, we further explore k-rUCG model (Section~\ref{sec:size-krucg}), retains the essential structure of rUCGs.
	
	\section{\label{sec:level4}Decomposition of rUCG with Lower Depth}
	The previous section presented a gate-count-optimized decomposition of the rUCG and k-rUCG, resulting in an average depth of 2 and 4 per multi-controlled gate, respectively. However, circuits built using this method exhibit excessive depth. This depth is already superior to that achievable by decomposing each multi-controlled gate individually. Specifically, if a decomposition with $O(log n)$ depth is adopted for a multi-controlled gate, the rUCG approach achieves a reduction by a factor of $O(log n)$. Given that an n-qubit multi-controlled gate cannot be optimized to constant depth, the rUCG method is unambiguously superior.
	
	However, from the perspective of the rUCG's native decomposition structure, where each gate occupies a dedicated circuit layer, the depth remains sub-optimal. To achieve a lower-depth decomposition, it is necessary to apply multiple gates within the same layer. This can be accomplished by modifying the traversal through the control state space, enabling each step of the traversal to activate distinct control states at different positions simultaneously.
	
	Similar to Section \ref{sec:size-trans} , define a new sequence of circuit $S = \{S_0, S_1, \dots, S_{D-1}\}$ involving CNOT gates. To handle the case where each layer yields multiple control states, we introduce for each step $i$ a set of $P_i$ activation qubits 
	\begin{equation}	
		\bigl\{\ket{s_{p}^{(i)}} : p \in P_i \subseteq [n]\bigr\},
	\end{equation}
	with each $\ket{s_{p}^{(i)}}$ carries $\bm{c} \cdot \boldsymbol{\omega}_p^{(i)}$, and modify Eq.~\eqref{eq:rucg-step} to multiple CU
	\begin{equation}\label{eq:rucgd}
		\begin{aligned}
			{\text{multi}}\text{-CU}_i &: \ket{s^{(i)}}\ket t \to 
			\prod_{p \in P_i} 
			CU_{s_p^{(i)}}(Y_{\omega^{(i)}_p})
			\left(
			\ket{s^{(i)}} 
			\ket t
			\right),\\ 
			S_i &: \ket{s^{(i)}} \to \ket{s^{(i+1)}},
		\end{aligned}
	\end{equation}
	with the initial state \(\ket{s^{(0)}}=\ket{0^{\otimes n}}\) and final state \(\ket{s^{(N-1)}}=\ket{0^{\otimes n}}\) ensuring a closed traversal for each branch. Here, $s_p^{(i)}$ represents a activation qubit $q$ during $i$-th step $CU_{s_p^{(i)}}(Y_{\omega^{(i)}_p})$ is controlled gate of corresponding target operator.
	
	Control states at the same depth are called a \emph{control state layer}. When we reduce the number of control state layers, the circuit depth can be optimized.
	
	This section adapts the depth optimization technique of unitary diagonal operators \cite{DUO2023SunX} to achieve significant depth reduction for rUCG circuits.
	
	\subsection{{\label{sec:depth-duo}}Depth bound for Diagonal Unitary Operator}
	
	A diagonal unitary operator, referred to as $\Lambda_n(\bm\chi)$ in Eq.~\eqref{eq:duo}, can be decomposed into a sequence of \{CNOT, $R_z$\} gates. The asymptotically optimal depth has been established in \cite{DUO2023SunX}, and \cite{DUO2024ZhangS} also provides several general optimization methods. The following discussion will adopt a more concise description and reformulate the problem within rUCG.
	
	For any decomposition of $\Lambda_n(\bm\chi)$, the CNOT-based sub-circuit implements the $S$ sequence, while the $R_z$-based sub-circuit corresponds to the multi-CU sequence. Eq.~\eqref{eq:rucgd} can thus be rewritten in the form of alternating CU layers and control state traversal layers:
	\begin{equation}\label{eq:duod}
		\begin{aligned}
			{\text{multi}}\text{-Rz}_i &: \ket{s^{(i)}} \to 
			\prod_{p \in P_i} 
			\exp{(jY_{\omega^{(i)}_p})}
			\ket{s^{(i)}}, \\ 
			S_i &: \ket{s^{(i)}} \to \ket{s^{(i+1)}}.
		\end{aligned}
	\end{equation}
	
	For a general decomposition of $\Lambda_n(\bm\chi)$ within \{CNOT, $R_z$\}, as indicated by Eq.~\eqref{eq:duod}, the CNOT gate sequence traverse through all possible control states, and the Rz gate introduces a phase to each corresponding control state.
	That is, rewrite the circuit as
	\begin{equation}
		\Phi_0 \; S_0 \; \Phi_1 \; S_1 \; \Phi_2 \; S_2 \; \cdots
	\end{equation}
	where each \(S_j\) is a CNOT circuit, and each \(\Phi_i\) is the a layer of Rz rotations. Let $D_{\mathrm{CNOT}}$ and $D_z$ represent the total depths of the two sequences, respectively. Here, $D_z$ is the number of $\{\Phi\}$ sequences, and $D_{\mathrm{CNOT}}$ is given by total depth of all $S$ sequences.
	
	\begin{lemma}
		\label{lem:duo'}
		For $\Lambda_n(\bm\chi)$, there exists a decomposition such that the depth consist of CNOT and Rz layers satisfies:
		\begin{equation}
			\label{eq:duo'}
			D'_{\mathrm{CNOT}}(n)
			= D'_{z}(n) O(\frac{n} {\log n}), \quad
			D'_{z}(n) \leq \frac{2^n+2}{n+1}-1,
		\end{equation}
		without ancilla.
	\end{lemma}
	
	\begin{proof}
		This proof combines all generate stages in \cite{DUO2023SunX} (Section 5).
		\begin{enumerate}
			\item \textbf{Minimal Sequence $S$:} The non-zero control states are partitioned into linearly independent subsets. \textbf{Linear independence} ensures that state can be achieved through CNOT networks. This step divides $\Lambda_n(\bm\chi)$ into $D'_{z}(n)$ control state layers, whose upper bound is proved by \cite{DUO2023SunX} (Appendix H).
			\item \textbf{Generate Stage:} For each divided control state layer,
			traversal depth of $S_i$ is $O(n/\log n)$, proved by \cite{DepthGenerate2020JiangJ} (Theorem 1).
			\item \textbf{Apply a rotation layer:} For each control state layer,
			apply a layer of Rz to achieve phases on each control state.
		\end{enumerate}
	\end{proof}
	
	\begin{theorem}[Depth of Diagonal Unitary Operators]
		\label{thm:duo-depth}
		For $\Lambda_n(\bm\chi)$,  its circuit depth consist of CNOT and Rz layers can be reduced to:
		\begin{equation} \label {eq: duo-bound}
			D_{\mathrm{CNOT}}(n)=
			D_{z}(n)=O\left({2^n}/n\right),
		\end{equation}
		without ancilla, which matches the asymptotic lower bound.
	\end{theorem}
	
	\begin{proof}
		\begin{figure}[tb]
			\centering
			\includegraphics[width=\columnwidth]{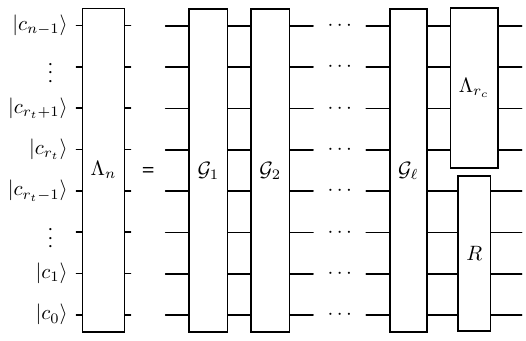}
			\caption{A circuit framework to implement a n-qubit unitary diagonal matrix, $\Lambda_{n}$, where 
				$r_t = \lfloor n/2 \rfloor$, $r_c = n - r_t = \lceil n/2 \rceil$, and $\ell_{k_t} \leq {2\binom{r_t}{k_t}}/{r_t}+2$.
				The depth of the operator $\mathcal{G}_{k_t}$ is $O(\frac 1n\binom n{k_t})$ for each ${k_t} \in [k]$, and the depth of the operator $R$ is $O(r_t / \log r_t)$.}
			\label{fig:duod}
		\end{figure}
		
		The implementation of \cite{DUO2023SunX} employs a recursive divide-and-conquer strategy, partitioning the n-qubit system into \textbf{prefix qubits} $\ket{c_c}$  of size $r_c = \lceil n/2 \rceil$ and \textbf{suffix qubits} $\ket{c_t}$  of size $r_t = \lfloor n/2 \rfloor$. The diagonal unitary $\Lambda_n$ is realized through the composition of the following steps:
		\begin{itemize}
			\item \textbf{Generate Stage:} Apply Generate State on $\ket{c_t}$ , denote as $S_i^{(c_t)}$, to traverse all non-zero states in suffix qubits.
			\item \textbf{Gray Path Stage:} For each arrived suffix state, 
			a Gray Path stage systematically extends control traversal from suffix qubits to all, with rotations applied for each state.
			From figure \ref{fig:GP}(a), an extra qubit can traversal with all state in control qubits within GP$(n)$.
			Figure \ref{fig:k-GP} demonstrates that arranging the gates in a pipeline can traverse concurrently within multiple qubits without increasing the overall circuit depth.
			The stage is called Gray Path Stage. A Gray Path Stage use $2^{r_c} $ CNOTs and $2^{r_c}$ Rz gates.
			\item \textbf{Iteration until Full Traversal:} Repeat the Generate Stage and the Gray Path Stage for $D'_z({r_t})$ times until all non-zero control states are completely traversed.
			\item \textbf{Suffix qubits Reset:} A linear reversible transformation $R$ is applied to reset suffix qubits to its initial state.
			\item \textbf{Recursive Invocation on the Prefix qubits:} Finally, the procedure \textbf{recursively executes} $\Lambda_{r_c}$ on prefix qubits to accurately apply the remaining phases within zero suffix.
		\end{itemize}
		
		Since the final two steps operate on distinct qubits, they can be executed in parallel. We take the larger one $\Lambda_{r_c}$ for depth calculation. The circuit depth satisfies
		

		\begin{equation}
			\begin{aligned} \label{eq:duo-complexity2}
				D_{z}(n) 
				&=   D'_z({r_t}) + 
				2^{r_c} D'_z({r_t}) +
				D_{z}({r_c}) \\
				~& = D_{z}({r_c}) +  (2^{r_c} + 1) D'_z({r_t})\\
				~& \leq D_{z}({r_c}) +  (2^{r_c} + 1) \left(\frac{2^{r_t}+2}{r_t+1}-1\right) ,
			\end{aligned}
		\end{equation}
		and 
		\begin{equation}
			\begin{aligned} \label{eq:duo-complexity1}
				D_{\mathrm{CNOT}}(n) 
				&=  D'_{\mathrm{CNOT}}({r_t}) 
				{+} 2^{r_c} D'_z({r_t}) 
				{+} D_{\mathrm{CNOT}}({r_c})	   \\
				~& = D_{\mathrm{CNOT}}({r_c}) 
				{+}  (2^{r_c} {+} O(\frac{n}{\log n})) D'_z({r_t}),\\
			\end{aligned}
		\end{equation}
		which resolves to the bound in Eq.~\eqref{eq: duo-bound}.
	\end{proof}
	
	\subsection{\label{sec:level4.2}Depth bound for k-sparse Diagonal Unitary Operator}
	\label{sec:depth-kduo}
	\begin{figure*}[tb]
		\centering
		\includegraphics[width=\linewidth]{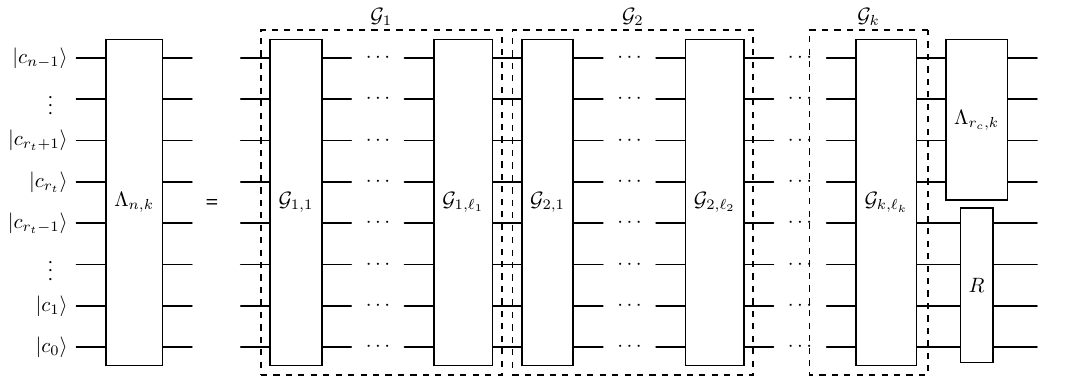}
		\caption{A circuit framework to implement a k-sparse n-qubit unitary diagonal matrix, $\Lambda_{n, k}$, where 
			$r_t = \lfloor n/2 \rfloor$, $r_c = n - r_t = \lceil n/2 \rceil$, and $\ell_{k_t} \leq {2\binom{r_t}{k_t}}/{r_t}+2$.
			The depth of the operator $\mathcal{G}_{k_t}$ is $O(\frac 1n\binom n{k_t})$ for each ${k_t} \in [k]$, and the depth of the operator $R$ is $O(r_t / \log r_t)$.}
		\label{fig:k-duo}
	\end{figure*}
	
	\begin{figure}[!t]
		\centering
		\includegraphics[width=\linewidth]{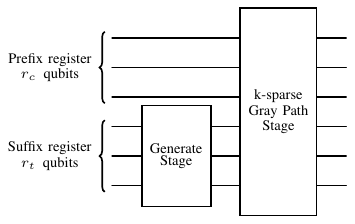}
		\caption{Implementation of operator $\mathcal{G}_{k_t, i}$. The depth of Generate Stage is $O(r_t / \log r_t)$, and the depth of k-sparse Gray Path Stage is $O(\sum_{k_c}\binom{r_c}{k_c}).$}
		\label{fig:G}
	\end{figure}
	
	We first present the depth bound for the decomposition of
	k-sparse diagonal unitary operators. The theorem summarizes the 
	depth scaling of the construction depicted in 
	figure \ref{fig:k-duo} and \ref{fig:G}. The decomposition follows the same high-level structure as Section~\ref{sec:depth-duo}, except that the suffix register must traverse only constant-weight subspace.
	
	Let $n = r_c + r_t$ be a qubit partition and decompose the weight as $k \geq k_c + k_t$.
	For each possible $k_t$, we construct a traversal layer $\mathcal{G}_{k_t}$ over all
	suffix states of weight $k_t$:
	\begin{equation}\label{eq:k-duo-restate}
		\mathrm{L\Lambda}_{n, k}
		= (\Lambda_{r_c, k}\otimes R)\prod_{k_t = k}^{1} \mathcal{G}_{k_t}.
	\end{equation}
	
	In the construction of $\mathrm{L\Lambda}_{n, k}$, modules R and Generate Stage are the same as Section~\ref{sec:depth-duo}. The remaining different modules are the main task of this section, including: creating Minimal Sequence $S$ to determine $D_z'(n,k)$, a k-sparse variant of Eq.~\eqref{eq:duo'}, and analyzing the complexity of k-sparse Gray Path Stage.
	
	Let
	\begin{equation}
		W_{k}^{(n)} = \left\{ 
		x\in\{0,1\}^{n} ~|~ \mathrm{wt}(x) = k
		\right\},|W_{k}^{(n)}| {=} \binom{n}{k}
	\end{equation}
	be all binary vectors of constant k-weight.

	\begin{lemma}
		\label{lem:min-group}
		$W_{k}^{(n)}$ can be partitioned into disjoint sets
		$G_1,$ $G_2,$ $\dots,$ $G_{\ell_{k_t}}$ such that each $G_i$ is linearly independent over $\mathbb{F}_2$, where $\ell_{k_t}$ denote the minimal number of linearly independent groups needed for
		$W_{k}^{(n)}$. Then,
		\begin{equation}\label{eq:G-bound-cor}
			\ell_{k}\le 
			\left\lceil
			\frac{\binom{n}{k}}{(n+e_{k})/2}
			\right\rceil,
		\end{equation}
		where $e_k=0$ if and only of $e_k$ is even, else $e_k=0$.  
	\end{lemma}
	
	\begin{proof}
		The detailed proof is deferred to Appendix~B for clarity in the supplementary material.
	\end{proof}
	
	\begin{lemma}
		For $\tilde{k}$-weight diagonal unitary $\Lambda_{n,\tilde{k}}^{=\tilde{k}}(\bm\chi)$ (distinct to k-sparse notation), there exists a decomposition such that the depth consist of CNOT and Rz layers satisfies:
		\begin{equation}
			\begin{aligned} 
				\label{eq:kduo'}
				D'_{\mathrm{CNOT}}(n,\tilde{k})
				&= D'_{z}(n,\tilde{k}) O(\frac{n} {\log n}), \\
				D'_{z}(n,\tilde{k}) 
				&= \ell_{\tilde{k}},
			\end{aligned}
		\end{equation}
		without ancilla.
	\end{lemma}
	
	\begin{proof}
		The proof is identical to that of Lemma ~\ref{lem:duo'}, except that Lemma ~\ref{lem:min-group} is used for the maximum independent set partitioning.
	\end{proof}
	
	\begin{lemma}
		\label{lem:comb-gray-depth}
		A k-sparse Gray Path stage that can be applied to traverses all $(k-k_t)$-sparse prefix configurations in parallel to all $k_t$-length suffix qubits. The stage can be implemented with depth less than
		\begin{equation*}
			3\sum_{k_c=0}^{k-k_t}\binom{r_c}{k_c}.
		\end{equation*} 
	\end{lemma}
	
	\begin{figure}[!t]
		\centering
		\begin{minipage}[b]{\linewidth}
			\centering
			\includegraphics[width=\linewidth]{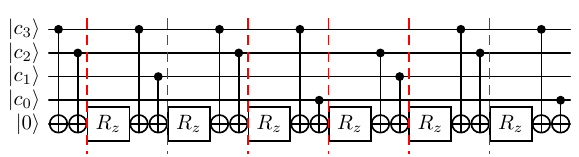}
			\caption*{(a)}
		\end{minipage}
		\begin{minipage}[b]{\linewidth}
			\centering
			\includegraphics[width=\linewidth]{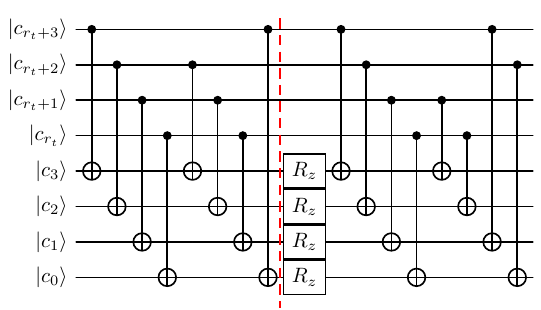}
			\caption*{(b)}
		\end{minipage}
		\caption{(a) k-sparse Gray Path GP$(n,k)$ with $n=4$ and $k=2$, using Rz as a inserted placeholder. (b) k-sparse Gray Path Stage  with $r_c=r_t=4$ and $k_c \leq 2$, which extends first two slice of (a) with combinatorial control state traversal to all suffix qubits without increasing depth. Then apply Rz gates to corresponding control states.}
		\label{fig:k-GP}
	\end{figure}
	
	\begin{proof}
		For each fixed suffix configuration of weight $k_t$, the remaining prefix is $k_c$-sparse.  
		To exhaust all prefix configurations, we employ k-sparse Gray Path GP$(n,k)$ in Section~\ref{sec:size-krucg}.
		
		Furthermore, prefix transitions can be executed \emph{in parallel} for every suffix qubit without
		increasing circuit depth, as illustrated in figure \ref{fig:k-GP}(b).
		
		Thus, the depth of k-sparse Gray Path Stage is equivalent to the depth of a GP$(r_c,k-k_t)$.
	\end{proof}
	
	\begin{theorem}[Depth of k-sparse diagonal unitary operators]
		\label{thm:ksparse-depth}
		Denote an n-qubit k-sparse diagonal unitary as $\Lambda_{n, k}(\boldsymbol{\chi})$.
		Then it admits a decomposition of the form in figure \ref{fig:k-duo}, with
		a total depth 
		\begin{equation}
			\label{eq:kduo-depth}
			D_{\mathrm{CNOT}}(n,k) = D_{z}(n,k)= O\!\left(\frac{1}{n}\sum_{i=1}^{k} \binom{n}{i}\right).
		\end{equation}
	\end{theorem}
	
	\begin{proof}
		Each layer $\mathcal{G}_{k_t}$ must traverse all $k_t$-weight suffix states.
		According to Lemma ~\ref{lem:min-group}, ~\ref{lem:comb-gray-depth} and figure \ref{fig:G},	summing over all $k_t$, the depth can be deducted as
		\begin{equation}\label{eq:kduo-complexity2}
			\begin{aligned} 
				D_{z}(n,k) 
				& \leq   D'_z({r_t,k}) {+}
				3\sum_{k_t=1}^{k} l_{k_t}\binom{r_c}{k{-}k_t} {+}
				D_{z}({r_c,k}) \\
				~& = D_{z}({r_c,k})
				+  \sum_{k_t=1}^{k} \left(3\binom{r_c}{k{-}k_t}+1\right)l_{k_t} ,
			\end{aligned}
		\end{equation}
		and 
		\begin{equation}\label{eq:kduo-complexity1}
			\begin{aligned} 
				&D_{\mathrm{CNOT}}(n,k) \\
				& \leq   D'_{\mathrm{CNOT}}({r_t,k}) {+}
				3\sum_{k_t=1}^{k} l_{k_t}
				\binom{r_c}{k{-}k_t} {+}
				D_{\mathrm{CNOT}}({r_c,k}) \\
				~& = D_{\mathrm{CNOT}}({r_c,k})
				{+}  \sum_{k_t=1}^{k} 
				\left(
				3\binom{r_c}{k{-}k_t} {+} O(\frac{n}{\log n})
				\right)l_{k_t} .
			\end{aligned}
		\end{equation}
		
		Finally, Eq.~\eqref{eq:kduo-depth} is proved.
	\end{proof}
	
	\subsection{Depth Bound for rUCG and k-rUCG Targeting on Rz}
	\label{sec:depth-ucrz}
	
	When the target operator of an rUCG or k-rUCG is a 
	rotation $R_z(\theta)$, the decomposition becomes 
	significantly simpler.
	Every controlled-$R_z$ gate is diagonal, and their product directly forms an 
	$(n{+}1)$-qubit diagonal unitary.
	Thus, the depth of such rUCGs is determined entirely by the known synthesis 
	results for diagonal unitary operators.
	
	\begin{theorem}[Depth of rUCGs targeting on Rz]
		\label{thm:rz-rucg}
		Denote an n-qubit rUCGs targeting on Rz as $\mathrm{LRz}_n(\mathbf\chi)$,
		which is equivalent to an $(n{+}1)$-qubit diagonal unitary, its optimal circuit depth 
		satisfies
		\begin{equation*}
			D_{\mathrm{CNOT}}(n+1)=
			D_{z}(n+1)=
			O\!\left(\frac{2^{n}}{n}\right).
		\end{equation*}
	\end{theorem}
	
	\begin{theorem}[Depth of k-rUCGs targeting on Rz]
		\label{thm:rz-k-rucg}
		Denote an n-qubit k-sparse rUCG targeting Rz as $\mathrm{LRz}_{n, k}(\mathbf\chi)$.
		Then its decomposition corresponds to an $(n{+}1)$-qubit k-sparse diagonal 
		unitary with depth
		\begin{equation*}
			D_{\mathrm{CNOT}}(n+1,k+1){=}
			D_{z}(n+1,k+1)
			{=} O\!\left(\frac{1}{n}\sum_{i=1}^{k} \binom{n}{i}\right).
		\end{equation*}
	\end{theorem}
	
	These bounds also apply to target rotations which are Clifford-equivalent to $R_z$, such as 
	$R_x$ and $R_y$, up to a constant number of local Clifford gates:
	\begin{equation*}
		R_x(\theta)=H\, R_z(\theta)\,H,
		\qquad
		R_y(\theta)=S^\dagger H\, R_z(\theta)\,H S.
	\end{equation*}
	
	\subsection{Shared-Target Controlled Gates}
	\label{sec:depth-stcg}
	\begin{figure}[!t]
		\centering
		\begin{minipage}[b]{\linewidth}
			\centering
			\includegraphics[width=\linewidth]{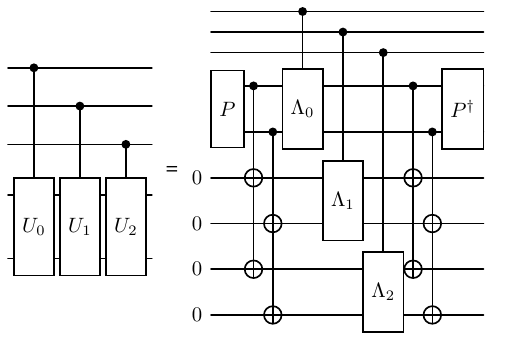}
			\caption*{(a)}
			\label{fig:stcg1}
		\end{minipage}
		\begin{minipage}[b]{\linewidth}
			\centering
			\includegraphics[width=\linewidth]{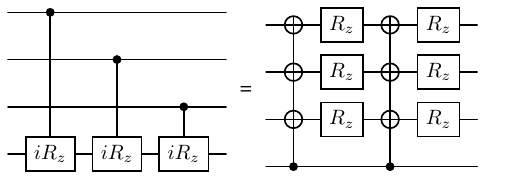}
			\caption*{(b)}
			\label{fig:stcg2}
		\end{minipage}
		\caption{Depth optimize for share-target controlled gates, with target operators commutative. (a) For controlled any $U$. \cite{STCU2001GreenF}  (b) For controlled single-qubit diagonal, namely controlled-$iR_z$. \cite{STCRz2025ZiW} Note that CNOTs with the same control is called quantum fan-out gate\cite{Fanout2021TakahashiY}, and the depth is $\log n$ for n CNOTs.}
		\label{fig:stcg}
	\end{figure}
	
	Diagonal unitary and rUCG share the same control traversal sequence and circuit structure, which also applies to the k-sparse version. As shown in Eq.~\eqref{eq:rucgd} and \eqref{eq:duod}, the construction of a diagonal unitary corresponds to an rUCG in which each controlled gate $CU(Y_{\omega})$ is replaced by an single-qubit phase $\exp(jY_{\omega})$.  So the implementation of rUCG can rely on the previous of diagonal unitary operator.
	
	The core issue is that $R_z$ gates, even when applied to different qubits, can all be executed in a single quantum layer. In contrast, CU gates cannot be executed in parallel when they share the same target qubits. As defined in Section~\ref{sec:rucg-represent}, the controlled-$U(\cdot)$ targets a gate set $\mathcal{U}$ endowed with the structure of a 2-divisible Abelian group.
	
	\begin{lemma}[Depth bound for Shared-Target Controlled Gates]
		\label{lemma:stcg}
		A sequence of controlled-$U(\cdot)$ gates with shared $m$ target qubits can be implemented in depth of $O(\log n)$ CNOTs and O(1) U-close gates, using $O(mn)$ ancillae.
	\end{lemma}
	
	\begin{proof}
		See figure \ref{fig:stcg}(a). Note that all $U(\cdot)$ can be diagonalized simultaneously, since $\mathcal{U}$ is Abelian.
	\end{proof}
	
	Furthermore, When the target operators follow stricter constraints, ancilla may not be required.
	
	\begin{definition}
		$\mathcal{U}$ is  said to admit a simultaneous diagonal–tensor decomposition if there exist fixed unitary A and B such that:
		\begin{equation*}
			U(\alpha)= A \bigotimes_{i=0}^{n-1} e^{j\theta_j(\alpha)Z} B,
		\end{equation*}
		where $e^{j\theta_j(\alpha)Z}$ is a single-qubit diagonal.
	\end{definition}
	
	\begin{lemma}
		\label{lemma:stcgv}
		If $\mathcal{U}$ admits a simultaneous diagonal–tensor decomposition with unitary A, B, a sequence of controlled-$U(\cdot)$ gates with shared $m$ target qubits can be implemented in without ancillae.
	\end{lemma}
	
	\begin{proof}
		A controlled single-qubit diagonal can be implemented by a controlled-Rz gate and a controlled-phase(equivalent to a Rz gate).
		As shown in \ref{fig:stcg}(b), Depth optimize for shared-target controlled single-qubit diagonal do not need ancilla.
	\end{proof}
	
	\subsection{Depth Bounds for rUCGs and k-rUCGs}
	\label{sec:depth-rucg}
	The shared-target analysis in the previous subsection allows us to
	translate the low-depth architecture of diagonal-unitary circuits into
	the setting of rUCGs.  We now
	formalize the resulting depth bounds for rUCGs and k-rUCGs.
	
	All results in this subsection follow the same structural template.
	The control traversal sequence $S$ contributes the CNOT depth, while
	the controlled-$U$ layer following each step contributes the CU depth.
	The only difference among the four settings is the degree of
	parallelizability of these latter gates.
	
	\begin{theorem}[Depth of rUCGs]
		\label{thm:rucg-depth}
		Denote an rUCG as $\mathrm{LU}_n(\boldsymbol{\chi})$. Then the decomposition achieves
		\begin{equation*}
			D_{\mathrm{CNOT}} = O(2^n / n),\quad
			D_{\mathrm{CU}} = O(2^n \log n / n),
		\end{equation*}
		using at most $O(mn)$ ancillae.
	\end{theorem}
	
	\begin{proof}
		The depth of CNOT is the same as that of Theorem ~\ref{thm:duo-depth}. The depth of CU needs to be multiplied by a coefficient of $log n$ due to Lemma ~\ref{lemma:stcg} and Lemma ~\ref{lemma:stcgv}.
	\end{proof}
	
	\begin{theorem}[Depth of k-rUCGs]
		\label{thm:krucg-depth}
		Denote a k-rUCG as $\mathrm{LU}_{n, k}(\boldsymbol{\chi})$.  Then the decomposition achieves
		\begin{equation*}
			D_{\mathrm{CNOT}}
			= O\!\left(\frac{1}{n}\sum_{i=0}^{k}\binom{n}{i}\right), 
			D_{\mathrm{CU}}
			= O\!\left(\frac{\log n}{n}\sum_{i=0}^{k}\binom{n}{i}\right),
		\end{equation*}
		using at most $O(mn)$ ancillae.
	\end{theorem}
	
	\begin{proof}
		The depth of CNOT is the same as that of Theorem ~\ref{thm:ksparse-depth}. The depth of CU needs to be multiplied by a coefficient of $\log n$ due to Lemma ~\ref{lemma:stcg} and Lemma ~\ref{lemma:stcgv}.
	\end{proof}
	
	\subsection{Section Discussion}
	\label{sec:depth-summary}
	This section addresses circuit depth of the rUCG decomposition framework. We develop a new depth-optimized decomposition for rUCGs in Section~\ref{sec:depth-rucg}, by systematically adapting the low-depth architecture of diagonal unitary operators in Section~\ref{sec:depth-duo}. Extending to the k-sparse version, we optimize the depth of k-sparse diagonal unitaries in Section~\ref{sec:depth-kduo}, which can be adapted to solve k-rUCG in Section~\ref{sec:depth-rucg}.
	
	Section~\ref{sec:depth-ucrz} and \ref{sec:depth-stcg} discuss different objective operators, including Rz, arbitrary unitaries, and unitaries with simultaneous diagonal–tensor decomposition. In practical applications, only the deep optimization of share target controlled gate cannot be directly solved by the rUCG model; targeted research is needed based on the specific gate.
	
	\section{\label{sec:level5}Examples}
	As a model, rUCG (including k-rUCG) provides optimization for a wide variety of quantum circuit types. 
	In this section, we an example to help understand how to apply our decomposition framework to optimize quantum circuits.
	
	\begin{figure}[!t]
		\centering
		\def\underlayer#1{\gategroup[4,steps=1,style={draw=none}, label style={label position=below,yshift=-0.2cm}]{#1}}
		\def\underlayertwo#1{\gategroup[4,steps=2,style={dashed, inner sep=0pt}, label style={label position=below,yshift=-0.35cm}]{#1}}
		\begin{minipage}[b]{\linewidth}
			\centering
			\includegraphics[width=\linewidth]{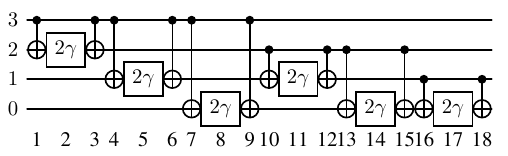}
			\caption*{(a)}
			\label{fig:qaoa1}
		\end{minipage}
		\begin{minipage}[b]{\linewidth}
			\centering
			\includegraphics[width=\linewidth]{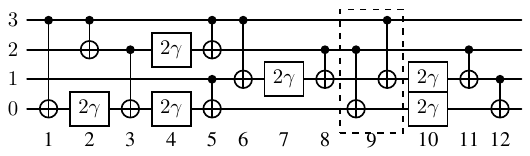}
			\caption*{(b)}
			\label{fig:qaoa2}
		\end{minipage}
		\begin{minipage}[b]{\linewidth}
			\centering
			\includegraphics[width=\linewidth]{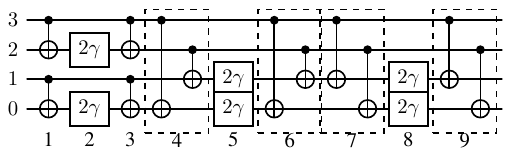}
			\caption*{(c)}
			\label{fig:qaoa3}
		\end{minipage}
		\caption{QAOA sub-circuit on 4 qubits (a) depth 18 by \cite{QAOACircuit2021LykovD}. (b) depth 12 by \cite{DUO2024ZhangS}. (c) depth 9 by our method. $R_z(-2\gamma)$ is marked as $2\gamma$ for short.}
		\label{fig:qaoa}
	\end{figure}
	The quantum approximate optimization algorithm (QAOA) \cite{QAOA2014FarhiE,QAOA2022BhartiK} has emerged as a promising approach for solving combinatorial optimization problems in the NISQ era, with its effectiveness partially hinging on the circuit depth required for implementing its ansatz. In particular, QAOA instances defined on complete graphs pose a significant challenge, as the corresponding diagonal operators typically require a quadratic number of two-qubit entangling gates. While existing synthesis techniques such as the Walsh-based or Gray-code-based constructions yield asymptotically optimal gate counts for diagonal unitary operators, the resulting circuit depth remains a bottleneck for near-term devices with limited coherence time.
	
	For a QAOA circuit defined on a complete graph\cite{QAOACircuit2021LykovD}, we investigate the sub-circuit between Hadamard and $R_x$ gates, which refers to a $\Lambda_{n,2}$ defined in Eq.~\eqref{eq:k-duo-restate}:
	
	\begin{equation}\label{eq:qaoa}
		\Lambda_{n,2}(\bm{\chi}) = \sum_{c=0}^{N-1} e^{j\chi_c}\ket c\bra c,
	\end{equation}
	where $\chi_c=\gamma\sum_{0\leq i_1< i_2< n}(-1)^{c_{i_1}\oplus c_{i_2}}$, which can be transformed into
	\begin{equation}\label{eq:qaoa2}
		\chi_c = \gamma\sum_{\substack{\omega\in[N]\\\mathrm{wt}(\omega)=2}} (-1)^{\vdot{c}{\omega}}.
	\end{equation}
	
	Then we apply circuit optimization for k-rUCG in Section~\ref{sec:level4.2} by the following steps.
	
	(1) Confirm that the k-sparse diagonal unitary operator form a k-sparse diagonal unitary operator. 
	
	(2) Term $\bm{\chi}$ as the target vector, and apply transform in Section~\ref{sec:size-trans} on it. For Eq.~\eqref{eq:qaoa2}, we get $X_{\omega}=N\gamma$, then derive $Y_{\omega}=-2\gamma$ .
	
	(3) Plan 2-weight control state traversal. $\mathcal{G}_2$ becomes a $\Lambda_{r_t,2}$ on suffix qubits, and $\mathcal{G}_1$ does not require a control state traversal. By derivation, the depth of CNOT is no more than $2n$, with n layers of control states reserved.
	
	(4) Insert $R_z(-2\gamma)$ to each correspond control state.
	
	An example is shown in figure \ref{fig:qaoa}(c) when $n=4$. We obtain a circuit with a shorter depth of 9, compared to the original one of depth 18 in figure \ref{fig:qaoa}(a), and the optimized one of depth 12 in figure \ref{fig:qaoa}(b). figure \ref{fig:qaoa_line} reveals that
	our strategy reduces the circuit depth by an order of magnitude.
	
	\begin{figure}[!t]
		\centering
		\includegraphics[width=\linewidth]{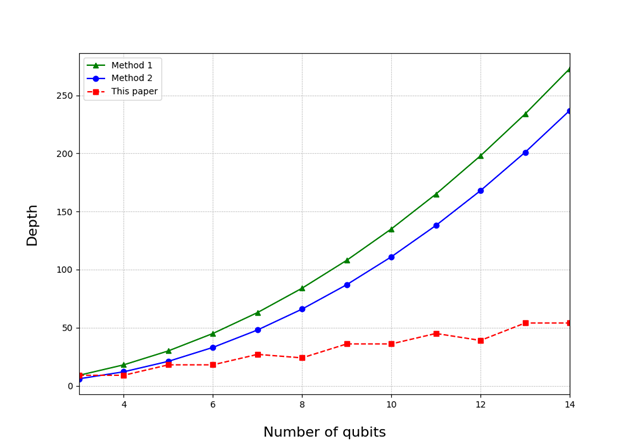}
		\caption{Depth comparison for three method to implement QAOA sub-circuit.}
		\label{fig:qaoa_line}
	\end{figure}

	\begin{table*}[htbp]
		\caption{
			Comparison of gate count, circuit depth, and ancilla requirements for various UCG decomposition methods.
		}
		\label{tab:ucg}
		\centering
		\begin{tabular}{lcccc}
			\toprule
			\textbf{Method} & \textbf{Gate Type} & \textbf{Gate Count} & \textbf{Circuit Depth} & \textbf{Ancilla} \\
			\midrule
			UCG (brute-force) with \cite{nCtrlDepth2024NieJ} 
			& Arbitrary $U$ 
			& $O(n2^n)$ basic,\; $O(2^n)$ CU 
			& $O(2^n \log n)$ basic,\; $O(2^n)$ CU 
			& 0 or 1 \\
			
			Uniformly controlled rotation \cite{UCR2004MöttönenM} 
			& $R_z$ or $R_y$ 
			& $O(2^n)$ basic
			& $O(2^n)$ basic 
			& 0 \\
			
			Diagonal unitary \cite{DUO2023SunX} 
			& Phase only 
			& $O(2^n)$ basic
			& $O(2^n/n)$ basic 
			& 0 \\
			
			rUCG (depth unoptimized)
			& Abelian, 2-divisible
			& $O(2^n)$ basic,\; $O(2^n)$ CU
			& $O(2^n)$ basic,\; $O(2^n)$ CU
			& 0 \\
			
			rUCG (depth optimized)
			& Abelian, 2-divisible
			& $O(2^n)$ basic,\; $O(2^n)$ CU
			& $O(2^n/n)$ basic,\; $O(2^n \log n / n)$ U-close
			& $[0,mn]$ \\
			\bottomrule
		\end{tabular}
	\end{table*}
	
	\begin{table*}[htbp]
		\caption{
			Comparison of k-rUCG decompositions.
			Here $\Sigma_k f(i)$ denotes $\sum_{i=0}^k f(i)\binom{n}{i}$.
		}
		\label{tab:k-ucg}
		\centering
		\begin{tabular}{lcccc}
			\toprule
			\textbf{Method} & \textbf{Gate Type} & \textbf{Gate Count} & \textbf{Circuit Depth} & \textbf{Ancilla} \\
			\midrule
			k-UCG (brute-force) \cite{nCtrlDepth2024NieJ} 
			& Arbitrary $U$
			& $O\!\left(\Sigma_k\, i\right)$ basic,\;
			$O\!\left(\Sigma_k\right)$ CU
			& $O\!\left(\Sigma_k \log i\right)$ basic,\;
			$O\!\left(\Sigma_k i\right)$ CU
			& 0 \\
			
			k-Diagonal unitary
			& Phase only
			& $O\!\left(\Sigma_k\right)$ basic
			& $O\!\left(\Sigma_k / n\right)$ basic
			& 0 \\
			
			k-rUCG
			& Abelian, 2-divisible
			& $O\!\left(\Sigma_k\right)$ basic,\;
			$O\!\left(\Sigma_k\right)$ CU
			& $O\!\left(\Sigma_k / n\right)$ basic,\;
			$O\!\left(\Sigma_k \log n / n\right)$ CU
			& $[0,mn]$ \\
			\bottomrule
		\end{tabular}
	\end{table*}
	%
	%
	%
	%
	%
	
	\section{{\label{sec:level7}}Conclusion}	
	This work proposes a unified and optimized model rUCG for the decomposition of UCGs. Decomposition framework based on rUCG yields improvements in both gate count and circuit depth.
	
	In terms of gate count, rUCGs can be implemented using both $O(2^n)$ CNOT and U-primitive gates. Further, k-rUCGs can be implemented using both $O\left(\sum_{i=0}^{k}\binom{n}{i}\right)$ CNOT and U-primitive gates. The best optimization was obtained without further decomposing the controlled operator.
	
	More significantly, this work presents a depth-optimized decomposition strategy. This reduces the CNOT depth from $O(2^n)$ to $O(2^n/n)$ and gate depth from $O(2^n)$ U-primitive gates to $O((2^n \log n)/n)$ U-close gates. Similarly, for the k-rUCG, gate count is reduced by a factor of n and depth are is reduced by $O(n/\log n)$.
	
	Our comparative analysis, summarized in Tables~\ref{tab:ucg} and~\ref{tab:k-ucg}, illustrates different algorithms applied to UCG or UCG-like operators. As demonstrated by the results, method proposed in this work consistently achieves the lowest gate count and circuit depth among all considered approaches.
	
	In addition to the theoretical framework, we have demonstrated the practical effectiveness of our techniques
	for QAOA circuits on complete graphs, we provide a depth reduction that scales quadratically compared to existing constructions. 
	
	Future Work. Two primary directions emerge from this study. First, a systematic investigation into how the specific algebraic properties of different gate types within the structured set $\mathcal{U}$ influence the decomposition complexity could refine our general framework. Second, exploring depth-optimization strategies through the strategic introduction of additional ancillary qubits presents another promising avenue for achieving more efficient circuit implementations.
	
	%
	
	%
	%
	%
	
	\section*{Acknowledgments}
	This work is supported by the Quantum Science and Technology-National Science and Technology Major Project (Grant No. 
	2021ZD0302901), National Natural Science Foundation of China (Grant No. 62071240) and Natural Science Foundation of Jiangsu 
	Province, China (Grant No. BK20220804).
	%
	%
	
	%
	\bibliographystyle{IEEEtran}
	\bibliography{a1}
	\appendices
	
	\section{Proof of Lemma~\ref{lem:k-sparse}: Control states of k-rUCG are all k-sparse}
	\label{app:k-sparse}
	
	This section proves that control states derived for the $\mathrm{LU}_{n,k}$ via transformation in Section~\ref{sec:size-trans} are all k-sparse. That is, for each $\wt(s^{(i)}_{\mathrm{a}}) \leq k$.
	
	Let's start with a single k-control gate, with $Q$ as control sets ($|Q| \le k$) , and $U(\mu)$ as target gate. Treat the gate as a k-rUCG, with target vector $\boldsymbol{\chi}$ and frequency vector $\boldsymbol{Y}$. According to Eq.~\eqref{eq:target-vector-kg} and proof of Theorem~\eqref{thm:krucg-subclass},
	\begin{equation*}
		\chi_i = 
		\begin{cases}
			\mu, & \text{if $Q \subset \supp(i)$},\\[2mm]
			0, & \text{otherwise},
		\end{cases}		
	\end{equation*}
	given $\supp(i)=\{t~|~i_t=1\}$.
	With at most $k$ bits fixed and at least $n-k$ bits free, there are exactly $2^{\,n-k}$ such $i$.
	
	A non-zero control state $s$ ($s\neq0$) is irrelevant if its associated controlled unitary is the identity $I$; equivalently, the corresponding entry of the frequency vector $Y_s$ is zero. Hence, to establish the claim, it suffices to show that any $s$ satisfies $\wt(s) > k$ necessarily yields a zero component in $Y_s$.
	
	According to Eq.~\eqref{eq:walsh-hadamard}\eqref{eq:frequecy-vector},
	\begin{equation*}
		Y_s=-\frac{2}{N}\sum_{i=0}^{N-1}(-1)^{\boldsymbol{s}\cdot\mathbf{i}}\,\chi_i,
	\end{equation*}
	where $\boldsymbol{s}\cdot\mathbf{i}=s_0i_0\oplus s_1i_1\oplus\cdots\oplus s_{n-1}i_{n-1}$.

	\subsection*{Step 1: Parameterizing all $i$ with $\chi_i=1$}
	
	If $\chi_i=1$, then $i_t=1$ for all $t\in Q$.  
	For the complementary index set $\overline Q$, the bits of $i$ are free.
	Thus every such $i$ can be uniquely written as
	\begin{equation*}
		i_t =
		\begin{cases}
			1, & t\in Q,\\
			y_j, & t\notin Q,
		\end{cases}
	\end{equation*}
	where $y\in\{0,1\}^{\,n-k}$ ranges freely.  
	Hence,
	
	\begin{equation*}
		Y_s=-\frac{2}{N}\sum_{\sum_{y\in\{0,1\}^{n{-}k}}}(-1)^{\boldsymbol{s}\cdot\mathbf{i}(y)}\,\chi_i,
	\end{equation*}
	where $\mathbf{i}(y)$ denotes $i$ with fixed bits 1s and free bits $y$.
	
	\subsection*{Step 2: Separating contributions from $Q$ and $\overline Q$}
	
	Decompose the mod-$2$ inner product:
	\begin{equation*}
		\boldsymbol{s}\cdot\mathbf{i}
		= \sum_{t\in Q} s_t\cdot 1 \;+\; \sum_{t\notin Q} s_t y_t
		= \Big(\sum_{t\in Q}s_t\Big) \;+\; (s_{\overline Q}\cdot y).
	\end{equation*}
	Therefore,
	\begin{equation*}
		(-1)^{\boldsymbol{s}\cdot\mathbf{i}}
		= (-1)^{\sum_{j\in Q}s_j}\,
		(-1)^{s_{\overline Q}\cdot y},
	\end{equation*}
	and so
	\begin{equation*}
		Y_s
		= -\frac2N(-1)^{\sum_{j\in Q}s_j} \cdot
		\sum_{y\in\{0,1\}^{n{-}k}}
		(-1)^{s_{\overline Q}\cdot y}.
	\end{equation*}
	
	\subsection*{Step 3: Evaluating the character sum}
	
	The inner sum factorizes coordinate-wise:
	\begin{equation*}
		\sum_{y\in\{0,1\}^{n-k}}
		(-1)^{s_{\overline Q}\cdot y}
		= \prod_{j\notin Q} \big( 1 + (-1)^{s_j} \big).
	\end{equation*}
	
	Each factor satisfies
	\begin{equation*}
		1 + (-1)^{s_j}
		=
		\begin{cases}
			2, & s_j=0,\\
			0, & s_j=1.
		\end{cases}
	\end{equation*}
	
	Thus, if there exists \(j\notin Q\) with \(s_j=1\), then one factor is \(0\) and the entire product equals \(0\). Meanwhile, $Y_s$ is then deduced to be 0.
	
	The condition $s_j=0$ for all $j\notin Q$ is equivalent to
	\begin{equation*}
		\operatorname{supp}(s)\subseteq Q
		\quad\Longrightarrow\quad
		\wt(s) \le |Q| \le k.
	\end{equation*}
	
	Therefore, whenever $\wt(s)>k$, the above condition cannot hold, and
	thus \(Y_s=0\).
	
	Back to real k-rUCGs, instead a single k-control gate, the target vector of the former is sum of the latter. Therefore, the conclusion whether a entry is 0 can be inherited.
	
	This completes the proof.
	\qed
	
	\section{Proof of Lemma~\ref{lem:min-group}: Construction of Maximum Independent Constant-Weight Prefix Sets}
	\label{app:min-group}
	
	This appendix provides the technical details underlying the grouping of constant-weight
	vectors in $\mathbb{F}_2^{n}$.
	
	\subsection{Modified Average Group Size}
	For
	\begin{equation*}
		W_{k}^{(n)} = \left\{ 
		x\in\{0,1\}^{n} ~|~ \mathrm{wt}(x) = k
		\right\},|W_{k}^{(n)}| {=} \binom{n}{k},
	\end{equation*}
	given the partition
	\(
	W_{k}^{(n)} = \bigcup_i G_i
	\)
	into linearly independent sets, the modified average group size $L_{n,k}$ is defined as
	\begin{equation*}
		L_{n,k}
		=\frac{\sum_{i\neq \text{shortest}} |G_i|}
		{(\text{number of groups})-1}.
	\end{equation*}
	This avoids instability caused by partially filled groups in recursive construction.
	Boundary values:
	\begin{equation*}
		L_{n,1}=n,\qquad L_{n,n-1}=n.
	\end{equation*}
	
	\subsection{Cross-Supplement Construction}
	Split the set by the last qubit:
	\begin{equation*}
		A=\{x0~|~x\in W_{k}^{(n-1)}\},\qquad
		B=\{y1~|~y\in W_{k-1}^{(n-1)}\}.
	\end{equation*}
	Inductively assume $A$ and $B$ have groupings with average sizes
	$L_{n-1,k}$ and $L_{n-1,k-1}$ respectively.
	
	Two supplement operations are performed:
	
	1) Move $B$ element to $A$ groups: since vector that end in 1 cannot be canceled out by any linear combination of vectors that end in 0, linear independence is preserved.
	
	2) \textbf{When $k$ is odd}, move $A$ element to $B$ groups: When ignoring the last bit, provided any odd-weight vector cannot lie in the span of even-weight vectors, linear independence is preserved too.
	
	Ignoring the shortest group yields the recurrence in the main text:
	\begin{equation*}
		L_{n,k}\ge
		\alpha_{n,k}(L_{n-1,k}+1)
		+\beta_{n,k}(L_{n-1,k-1}+e_{k}),
	\end{equation*}
	where $e_k=0$ if and only of $e_k$ is even, else $e_k=0$.
	
	\subsection{Lower Bound on group count}
	Induction on $n$ along the recurrence completes the derivation of the bound
	\begin{equation*}
		L_{n,k}\ge \frac{n+e_{k}}{2}.
	\end{equation*}
	Finally, the bound on the number of groups
	\begin{equation*}
		\ell_{k}\le 
		\left\lceil\frac{\binom{n}{k}}{(n+e_{k})/2}\right\rceil,
	\end{equation*}
	follows directly from the definition.
	
\end{document}